\documentclass[a4paper,12pt,showkeys]{article}
\usepackage{color}
\usepackage{amsmath,amssymb,amsthm}
\usepackage{afterpage}  
\usepackage{tcolorbox}
\usepackage{bm}
\usepackage{makeidx}
\usepackage{setspace}
\usepackage{graphicx}
\usepackage{mathtools}
\usepackage{comment}
\usepackage{amsfonts}
\usepackage{bbm,bm}
\usepackage{physics}
\usepackage{float}
\usepackage{cancel}
\usepackage{cases}
\usepackage{ascmac}
\usepackage{ulem} 
\usepackage[english]{babel}
\usepackage[hang,flushmargin]{footmisc}
\usepackage{stmaryrd}
\usepackage[top=33truemm,bottom=33truemm,left=27truemm,right=27truemm]{geometry}
\usepackage{setspace}
\usepackage{hyperref}
\usepackage{enumitem}
\usepackage{natbib}

\theoremstyle{plain}

\newtheorem{theorem}{Theorem}

\newtheorem{corollary}{Corollary}
\theoremstyle{definition}
\newtheorem{definition}{Definition}
\newtheorem{example}{Example}

\newtheorem{algorithm}{Algorithm}

\makeatletter
\renewenvironment{proof}[1][\proofname]{\par
  \pushQED{\qed}%
  \normalfont \topsep6\p@\@plus6\p@\relax
  \trivlist
  \item\relax
  {\bfseries
  #1\@addpunct{.}}\hspace\labelsep\ignorespaces
}{%
  \popQED\endtrivlist\@endpefalse
}
\makeatother

\usepackage{listings}
\usepackage[T1]{fontenc} 
\usepackage[framed,numbered]{matlab-prettifier}
\definecolor{darkred}{RGB}{171, 14, 12}
\definecolor{main}{RGB}{0, 0, 0}
\color{main}

\title{Matching with regional constraints:
An equivalence\footnote{We express sincere gratitude to my advisor Fuhito Kojima for encouraging guidance, thoughtful support and detailed suggestions. We are especially grateful for Michihiro Kandori for posing this question. We are grateful for In-Koo Cho, Vijay Krishna, Arunava Sen, Antonio Penta for extensive comments and suggestions. We thank Yuichiro Kamada, Akira Matsushita for helpful discussions. We thank Ayano Yago and conference audiences at the Summer School of Econometric Society in Abu Dhabi (2024) for helpful comments.} }
\author{Elizabeth Nanami Aoi\footnote{University of Tokyo, email: enanamiaoi@g.ecc.u-tokyo.ac.jp}}
\date{\today }
\begin{document}
\maketitle 

\providecommand{\keywords}[1]
{
  \small	
  \textbf{\textit{Keywords---}} #1
}

\providecommand{\jelcodes}[1]
{
  \small	
  \textbf{\textit{JEL---}} #1
}

\begin{abstract}
In two-sided matching market, when the regional constraints are present, the deferred acceptance (DA) algorithm suffers from undesirable inefficiency due to the artificial allocation of the regional caps among hospitals. We show that, given preferences, there exist allocations that guarantee the efficiency of the DA algorithm.  Furthermore, it is equivalent to the FDA algorithm developed by \cite{kamada2015efficient}, which endows the latter with an interpretation as a tool for endogenous capacity design. Our proof applies the optimality within the matching with contracts (\cite{hatfield2005matching}) framework, offering a broadly applicable method for establishing equivalence among DA-based mechanisms.

\end{abstract}


\jelcodes{C70, D61.}

\section{Introduction}

One of the crucial goals of matching theory is to find ``desirable'' matchings. In a two-sided matching market, the deferred acceptance (DA) algorithm has been studied and applied for its unique properties, with real-life contexts such as school choice (\cite{abdulkadirouglu2003school}), medical residency matching (\cite{roth1984evolution}, \cite{roth1999redesign}), and others. However, practical constraints often prevent the direct application of the DA algorithm, giving rise to the field known as ``matching with distributional constraints.''

For example, a central concern in medical residency matching is the regional disparity of doctors, which urges the government to limit placements in urban areas to support rural regions. Formally, each region has a regional cap, an upper bound on the number of doctors it can accept. \cite{kamada2015efficient} formalized this setting as ``matching with regional constraints''. Although the DA algorithm can be applied in such settings, it requires artificially tightening constraints. Because the DA algorithm operates in a standard two-sided market, the regional cap of each region must first be \textit{allocated} among hospitals that are located in that region to satisfy both regional caps and capacities of the hospitals. However, this artificial allocation of the regional caps sometimes causes problems. 

To illustrate this point, consider a region with only two hospitals: one highly preferred (`good') and one less preferred (`bad'). The regional cap and the capacity of each hospital are one. A doctor prefers the good hospital, but would rather remain unmatched than go to the bad one. If the regional cap is allocated to the bad hospital, the DA algorithm yields an empty match, wasting the regional cap. Instead, allocating to the good hospital results in an efficient match.

The above example highlights the remaining challenge in applying the DA algorithm: \textit{how to allocate regional caps to fully realize its potential}. While previous research did not directly address this allocation problem, \cite{kamada2015efficient} conducted a thorough investigation on matching with regional constraints and proposed the flexible deferred acceptance (FDA) algorithm, a new mechanism that is strategy-proof and produces a (constrained) efficient, weakly stable matching. 

This paper focuses on the problem of allocation of the regional caps and elucidates the equivalence between the FDA algorithm and the DA algorithm. Specifically, if regional caps are allocated to mimic the distribution produced by the FDA algorithm, then the DA algorithm yields exactly the same (constrained) efficient matching. This result gives a new interpretation to the FDA algorithm as a tool of endogenous capacity design for the DA algorithm. Our simple proof technique, which exploits (proposing-side) optimality, a celebrated property of the DA algorithm, is readily applicable to other cases in establishing the equivalence of the DA-based mechanisms. More precisely, taking advantage of the matching with contracts (\cite{hatfield2005matching}) framework as \cite{kamada2018stability} did in their proof, we relate the two algorithms to the generalized DA algorithm in two different markets and show that both algorithms yield stable matchings in both markets. Then, optimality in each market prompts equivalence.
Rationalization of choice functions makes the argument concise. 
Our paper contributes to the literature on matching with distributional constraints by offering a new perspective on the implementation of (constrained) efficient outcomes using familiar algorithms.

As the standard two-sided matching theory has matured, growing practical needs in applications have driven research toward environments with distributional constraints. Notable examples include diversity constraints in school choice\footnote{Such area of study is sometimes called as ``controlled school choice''.} such as maximum quotas (\cite{abdulkadirouglu2005college}; \cite{kojima2012school}), reserves (\cite{hafalir2013effective}), minimum quotas (\cite{fragiadakis2017improving}; \cite{tomoeda2018finding}), soft bounds (\cite{ehlers2014school}) and affirmative action policies in India (\cite{aygun2020dynamic}; \cite{sonmez2022affirmative}), as well as regional constraints in medical residency matching with floor constraints (\cite{akin2021matching}) and with regional caps (\cite{kamada2015efficient}; \cite{kamada2017stability}; \cite{goto2014improving}). \cite{kamada2018stability} further generalized the regional caps to hierarchical structures and proposed an extension of the FDA algorithm to ensure efficiency, a setting to which our results also directly apply.\footnote{See Section 4.} 

To analyze priority and constraints at the same, some papers take an approach to utilize choice rules (\cite{echenique2015control}: \cite{imamura2020meritocracy}; \cite{dougan2023choice}). In fact, in studying more general constraint frameworks, it is common to aggregate preferences of hospitals and distributional constraints into a single choice function, enabling analysis via the matching with contracts (\cite{hatfield2005matching}) framework. \cite{hafalir2022design} utilized utility rationalization of a choice function to pursue distributional objectives, focusing on lexicographic choice functions. \cite{kojima2018designing} proposed a unifying framework that subsumed many of the above specific models. Our results may extend to such settings when the hospital side prioritizes regional preferences over the preferences of hospitals.

The rest of this paper is organized as follows: Section 2 formally states the model of matching with regional constraints and describes the DA algorithm and the FDA algorithm. Section 3 presents our results and proofs, and Section 4 provides discussion and possible generalizations.

\section{Matching with regional constraints}

\subsection{Model}

There is a finite set of doctors $D$ and a finite set of hospitals $H$. Each doctor $d \in D$ has a strict preference relation $\succ_d$ over the set of hospitals and the state of being unmatched $\phi$. For any $h', h'' \in H \cup \{\phi\}$, we write $h' \succeq_d h''$ if and only if $h' \succ_d h''$ or $h' = h''$. Each hospital $h \in H$ has a strict preference relation $\succ_h$ over the set of subsets of doctors. For any $D', D'' \subseteq D\cup \{\phi\}$, we write $D' \succeq_h D''$ if and only if $D' \succ_h D''$ or $D' = D''$. We denote the preference profile of all doctors and hospitals as $\succ \equiv (\succ_i)_{i \in D \cup H}$. We say that the doctor $d$ is \textbf{acceptable} to hospital $h$ if $d \succ_h \phi$. Conversely, $h$ is acceptable to $d$ if $h \succ_d \phi$.

Each hospital $h \in H$ has a (physical) \textbf{capacity} $q_h$, which is a positive integer. We assume that the preference of each hospital $\succ_h$ is \textbf{responsive with capacity} $q_h$ (\cite{roth1985college}), 
\begin{enumerate}
    \item[(i)] For any $D' \subseteq D$ with $|D'| \leq q_h$, $d \in D\backslash D'$ and $d' \in D'$, $(D' \cup d) \backslash d' \succeq_h D'$ if and only if $d \succeq_h d'$,
    \item[(ii)] For any $D' \subseteq D$ with $|D'| \leq q_h$ and $d' \in D'$, $D'\succeq_h D' \backslash d'$ if and only if $d' \succeq_h \phi$, and
    \item[(iii)] $\phi \succ_h D'$ for any $D' \subseteq D$ with $|D'| > q_h$.
\end{enumerate}

There is a finite set of \textbf{regions} $R$, and the set of hospitals is partitioned into regions. That is, denoting the set of hospitals in region $r\in R$ as $H_r$, we have $H_r \cap H_{r'} = \phi$ for all pairs of two different regions $r, r' \in R$ and $\cup_{r \in R} H_r = H$. Each region $r \in R$ is endowed with a \textbf{regional cap} $q_r \in \mathbb{Z}_+$. For each $h \in H$, let $r(h)$ denote the region in which hospital $h$ resides, that is, $h \in H_{r(h)}$.

A \textbf{matching} $\mu$ is an assignment of doctors to hospitals. Formally, it is a mapping defined on $D \cup H \cup \{\phi\}$ that satisfies (i) $\mu_d \in H \cup \{\phi\}$ for all $d \in D$, (ii) $\mu_h \in 2^D \cup \{\phi\}$ for all $h \in H$ and (iii) for any $ d \in D$ and for any $h \in H$, $\mu_d = h$ if and only if $d \in \mu_h$.

A matching $\mu$ is \textbf{feasible} if (i)
$|\mu_r| \leq q_r$ for all $r \in R$, where $\mu_r \equiv \cup_{h \in r} \mu_h$  and (ii) $|\mu_h| \leq q_h$ for all $h \in H$. This definition requires compliance with the regional cap for every region as well as the capacities of the hospitals. This model allows for a situation where $q_r < \sum_{h \in H} q_h$, meaning that the regional cap can be smaller than the sum of hospital capacities in that region. 

The efficiency in this model is a constrained notion due to the regional caps. A feasible matching $\mu$ is \textbf{(constrained) efficient} if there exists no other feasible matching $\mu'$ such that $\mu'_i \succeq_i \mu_i$ for all $i \in D \cup H$ and $\mu'_i \succ_i \mu_i$ for some $i \in D \cup H$.

A matching is \textbf{individually rational} if no doctor or hospital is matched with an unacceptable agent. 

Given a matching $\mu$, a pair of a doctor and a hospital $(d, h) $ is a \textbf{blocking pair} if $h \succ_d \mu_d$ and either (i) $|\mu_h| < q_h$ and $d \succ_h \phi$, or (ii) $d \succ_h d'$ for some $d' \in \mu_h$. It is a pair who is eager to match each other on their own rather than following the proposed matching $\mu$. 

In a standard two-sided matching market without regional constraints, a matching $\mu$ is said to be \textbf{stable} if it is individually rational and it has no blocking pair. However, in the presence of regional constraints, stability is not always compatible. Hence, \cite{kamada2017stability} proposed an appropriate notion of fairness called weak stability, whose existence is always guaranteed.\footnote{In fact, \cite{kamada2017stability} investigated a more natural notion of stability in this setting called strong stability. However, they showed that strong stable matching does not always exists.}  A matching $\mu$ is said to be \textbf{weakly stable} if it is feasible, individually rational, and if $(d, h)$ is a blocking pair then (i) $\sum_{h' \in r(h)} |\mu_{h'}| = q_{r(h)}$, and (ii) $d' \succ_h d$ for all $d' \in \mu_h$. In words, the weak stability accommodates some blocking pair $(d, h) $, as long as the regional cap is binding for the region $h$ is located in, and all doctors matched with $h$ are strictly preferred to $d$.

A \textbf{mechanism} $\phi$ is a function that maps preference profiles to matchings.  We denote the matching produced by mechanism $\phi$ given preference profile $\succ$ as $\phi(\succ)$, and agent $i$'s match as $\phi_i(\succ)$ for each $i \in D\cup H$. A mechanism $\phi$ is said to be \textbf{strategy-proof for doctors} when no doctor has an incentive to report an untrue preference given the mechanism. Formally, $\phi$ is strategy-proof for doctors if there exists no preference profile $\succ$, a doctor $d \in D$, and some preference $\succ'_d$ of doctor $d$ such that 
\begin{align*}
    \phi_d(\succ'_d, \succ_{-d}) \succ_d \phi_d(\succ).
\end{align*}

\subsection{The DA algorithm and the allocation of the regional caps}

In this section, we describe the motivating real-life example of the matching with regional constraint, the medical residency matching in Japan. They use a mechanism called the Japanese Residency Matching Program (JRMP) mechanism, which is the deferred acceptance (DA) algorithm with an adapted capacities. The DA algorithm, which was introduced by \cite{gale1962college}, is defined in a standard matching market without regional caps. In our model, this corresponds to the case in which the regional caps cannot bind for all regions, that is, for each $r \in R$, we have $q_r > |D|$.

\begin{algorithm}
The \textbf{DA algorithm} proceeds in the following steps.\footnote{To be precise, this is called the doctor-proposing DA algorithm.}

\textbf{Step 1.} 
Each doctor applies to the hospital of her first choice. Given the applications, each hospital rejects the the lowest-ranking doctors in excess of its capacity and all unacceptable doctors, while temporarily keeps the rest. 

In general,

\textbf{Step k.} 
Each doctor who got rejected in the previous step applies to her next highest choice (if any). Each hospital combines these newly applied doctors and those who are temporarily kept in the previous step, and rejects the lowest-ranking doctors in excess of its capacity and all unacceptable doctors, while temporarily keeps the rest.

The algorithm terminates at a step in which no application occurs. It is well-known that this algorithm terminates in the finite steps.
\end{algorithm}

The DA algorithm is known to be strategy-proof for doctors, and always produces an efficient, stable matching. Furthermore, Hatfield and Milgrom (2005) showed that the DA algorithm produces the (proposing-side) optimal stable matching. A stable matching $\mu$ is said to be (proposing-side) \textbf{optimal} if there exists no stable matching $\mu'$ such that $\mu'_d \succeq_d \mu_d$ for all $d \in D$ and $\mu'_{d'} \succ_{d'} \mu_{d'}$ for at least some $d' \in D$. 

The JRMP mechanism was introduced in the hope of taking advantage of these many desiderata of the DA algorithm: In addition to hospital capacities and regional caps, the government exogenously imposes an artificial capacity for each hospital $h$, or \textbf{target capacity} $\bar{q}_h$, which is a positive integer, to satisfy the regional caps, i.e., $\sum_{h \in r} \bar{q}_h \leq q_r$ for each region $r \in R$. The JRMP mechanism is the DA algorithm in which the capacity of each hospital $h$ is replaced by the target capacity $\bar{q}_h$, instead of the (physical) capacity $q_h$. Simply put, \textit{they ``allocate'' the reginal caps among the hospitals in each region before applying the DA algorithm.}

However, this allocation of the regional caps often fails -- the JRMP mechanism sometimes lacks some desiderata the DA algorithm possesses, such as (constrained) efficiency. In fact, this comes from the failure in allocating the regional caps in the JRMP mechanism. Since it is the regional caps that the government wishes to be satisfied, not the target capacities, the allocation of the regional caps might pose an unnecessarily demanding constraint. The following example highlights the possible inefficiency of arbitrary allocation.

\begin{example}[Failure of the allocation of regional caps]
    Suppose there is one region $r$ with regional cap $q_r = 4$. In region $r$, there are three hospitals $h_1$, $h_2$ and $h_3$, with (physical) capacities $q_{h_1} = q_{h_2} = q_{h_3} =  2$. There are five doctors $d_1$, $d_2$, $d_3$, $d_4$ and $d_5$. The preference profile $\succ$ is as follows:
    \begin{align*}
        &\succ_{h_i}  : d_1 , d_2 , d_3, d_4, d_5, \phi \text{ for } i = 1, 2, 3,\\
        &\succ_{d_i}  : h_1, h_2, \phi \text{ for } i = 1, 2, 3, \\
        &\succ_{d_4}  : h_2, \phi,\\
        &\succ_{d_5} : h_2, h_3, \phi.
    \end{align*}
    
    Suppose that the government allocates the regional cap of $r$ among three hospitals as $(\bar{q}_{h_1}, \bar{q}_{h_2}, \bar{q}_{h_3}) = (1, 1, 2)$. Then, the DA algorithm produces the following matching
    \begin{align*}
        \mu = 
        \begin{pmatrix}
             h_1 & h_2 & h_3 & \phi \\
             d_1 & d_2 & d_5 & d_3, d_4\\
    \end{pmatrix},
    \end{align*}
    which means $d_1$ matches with $h_1$, $d_2$ matches with $h_2$, $d_5$ matches to $h_3$, and $d_4$ and $d_5$ remain unmatched. We use this matrix form to specify a matching in the following as well.
    It is obvious that $\mu$ is inefficient since it is Pareto dominated by the following matching $\mu'$, which is (constrained) efficient.
    \begin{align*}
        \mu' = 
        \begin{pmatrix}
             h_1 & h_2 & h_3 & \phi\\
             d_1 & d_2, d_3 & d_5 &  d_4  \\
    \end{pmatrix}.
    \end{align*} 
    Notice that, if the government instead adopts the allocation $(\bar{q}'_{h_1}, \bar{q}'_{h_2}, \bar{q}'_{h_3}) = (1, 2, 1)$, then the DA algorithm produces $\mu'$. 
\end{example}

We say that the allocation of the regional caps is efficient if the DA algorithm under that allocation produces a (constrained) efficient matching. In Example 1, $(\bar{q}'_{h_1}, \bar{q}'_{h_2}, \bar{q}'_{h_3}) = (1, 2, 1)$ is an efficient allocation, while $(\bar{q}_{h_1}, \bar{q}_{h_2}, \bar{q}_{h_3}) = (1, 1, 2)$ is not. This highlights that how to design the allocation of the regional caps is the primary determinant of the performance of the DA algorithm. One question will be whether efficient allocation is guaranteed to exist, to which we provide a positive answer in Section 3.

\subsection{The FDA algorithm}

In the previous section, we emphasize the importance of the desirable allocation of the regional caps. In this section, we see how previous literature addressed inefficiency of the existing mechanism. \cite{kamada2015efficient} tackled this problem by directly giving matchings that surmount the deficiencies through a new mechanism called the flexible deferred acceptance (FDA) algorithm. They showed that the FDA algorithm returns a (constrained) efficient matching while making no doctors worse off.

\begin{algorithm}
    In the \textbf{FDA algorithm} (Kamada and Kojima 2015), first, we specify an \textbf{order} among all hospitals in $H_r$, where $H_r$ denote the set of hospitals in region $r\in R$. That is, $H_r\equiv \{h_1, h_2, \ldots, h_{|H_r|}\}$, where ther order is specified to be $h_1, h_2, \ldots, h_{|H_r|}$. Beginning with an empty matching, $\mu_d = \phi$ for all $d \in D$, the algorithm proceeds as follows.
    
    At each \textbf{Step k}, pick one doctor $d$ who is not tentatively kept by any hospital. Let $d$ apply to $\bar{h}$, a hospital of her first choice among those who have not rejected her yet (if any). Let $r \in R$ denote the region $\bar{h}$ is located.

    \begin{itemize}
        \item[(i)] \textit{Filling the target capacities.} For each hospital $h \in H_r$, among all doctors who have applied to and not been rejected from $h$, tentatively keep from the highest-ranking doctors up to its target capacity. 
        \item[(ii)] \textit{Filling the regional cap respecting the order.} Start with the tentative match in (i). Hospitals in region $r$ take turns in the pre-specified order (starting from $h_1$, followed by $h_2$, $\ldots$, $h_{|H_r|}$, and go back to $h_1$ again) in additionally keeping the best remaining doctors until the regional cap $q_r$ is reached or its (phisical) capacity $q_h$ is filled or no doctor remains in the application pool. Reject the rest.
    \end{itemize}
    If there exists no further application, the algorithm terminates. 
\end{algorithm}

Notice that the FDA algorithm is not unique since its definition depends on the orders of hospitals in each region. However, regardless of the order, \textit{the FDA algorithm always returns efficient matchings}. The following example illustrates this.

\begin{example}[Efficiency of the FDA algorithm: \cite{kamada2015efficient}]
    Consider the setting in Example 1, and assume that the government allocates the regional caps of $r$ as $(\bar{q}_{h_1}, \bar{q}_{h_2}, \bar{q}_{h_3}) =(1, 1, 2)$. This time, we run the FDA algorithm. Remember that the physical capacities are $(q_{h_1}, q_{h_2}, q_{h_3}) = (2, 2, 2)$. 
    
    In the FDA algorithm, we first need to specify the order among hospitals in each region. We have $3!$ different orders, but given the preference profile of doctors, the relevant order is the order among $h_1$ and $h_2$. (This is because $h_3$ can receive at most one application (from $d_5$), which is smaller than $\bar{q}_{h_3} = 2$.) So, all the possible orders fall into two: (A) $h_1, h_2$ or (B) $h_2, h_1$.\footnote{(A) includes $h_1 , h_2, h_3$ or $h_1, h_3, h_2$ or $h_3, h_1 , h_2$, and (B) includes $h_2 , h_1, h_3$ or $h_2, h_3, h_1$ or $h_3, h_2 , h_1$} Suppose the FDA algorithm picks the order (A) $h_1, h_2$. Then it proceeds as follows. 
    
    \textbf{Step 1. }First, the doctors apply to their most favorite hospitals; $d_1$, $d_2$ and $d_3$ apply to $h_1$, and $d_4$, $d_5$ apply to $h_2$. 
    \begin{itemize}
        \item[(i)] $h_1$ and $h_2$ tentatively keep up to their target capacities $\bar{q}_{h_1} = \bar{q}_{h_2} = 1$. $h_1$ keeps $d_1$ and $h_2$ keeps $d_4$. The regional cap has two additional seats.
        \item[(ii)] Hospitals take turns in keeping additional doctors following the order $h_1, h_2$. $h_1$ keeps $d_2$, and $h_2$ keeps $d_5$. The remaining doctor $d_3$ gets rejected.
    \end{itemize}

    \textbf{Step 2.} $d_3$ applies to her next best hospital $h_2$.
    \begin{itemize}
        \item[(i)] $h_1$ and $h_2$ tentatively keep up to their target capacities. $h_1$ keeps $d_1$, and $h_2$ keeps $d_3$. The regional cap has two additional seat.
        \item[(ii)] Hospitals take turns in keeping additional doctors following the order $h_1, h_2$. $h_1$ keeps $d_2$, and $h_2$ keeps $d_4$. The remaining doctor $d_5$ gets rejected.
    \end{itemize}

    \textbf{Step 3. } $d_5$ applies to her next best hospital $h_3$. 
    \begin{itemize}
        \item[(i)] All three hospitals tentatively keep up to their capacities. $h_1$ keeps $d_1$, $h_2$ keeps $d_3$, and $h_3$ keeps $d_5$. The regional cap has one additional seat.
        \item[(ii)] Hospitals take turns in keeping additional doctors following the order $h_1, h_2$. $h_1$ keeps $d_2$. The remaining doctor $d_4$ gets rejected. There is no further application, the FDA algorithm terminates. 
    \end{itemize}
    
   We obtain the following matching $\mu''$, which is (constrained) efficient.
    \begin{align*}
        \mu'' = 
        \begin{pmatrix}
             h_1 & h_2 & h_3 & \phi \\
             d_1, d_2 & d_3 & d_5 & d_4  \\
        \end{pmatrix}.
    \end{align*} 

    If the FDA algorithm works with the order (B) $h_2, h_1$ instead, it returns $\mu'$ in Example 1, which is also (constrained) efficient. Regardless of the orders, the FDA mechanism returns (constrained) efficient matchings. 
\end{example}

The (constrained) efficiency of the FDA algorithm overcomes the deficiencies of the existing mechanism. However, the connection between the FDA algorithm and the allocation of the regional caps has yet to be investigated. Our main result, provided in the next section, gives a clear answer to this problem, entangling the relationships between the FDA algorithm and the DA algorithm. We prove that, not only does there exist an efficient allocation, it coincides with the distribution of doctors under some FDA algorithm. Furthermore, \textit{the output of the DA algorithm under this allocation of the regional caps is exactly the same as the outcome of the corresponding FDA algorithm.} We conclude this section with an illustrative example.

\begin{example}[The equivalence of the DA algorithm and the FDA algorithm]
    Let us refer to Example 1 again. Now, start with the allocation $(\bar{q}'_{h_1}, \bar{q}'_{h_2}, \bar{q}'_{h_3}) = (1, 2, 1)$, which mimics the distribution of $\mu'$, a matching produced by the FDA algorithm with order (A) $h_2, h_1$.
    
    The DA algorithm proceeds as follows. First, $d_1$, $d_2$ and $d_3$ apply to $h_1$ and $d_4$, $d_5$ apply to $h_2$. Then, $d_2$, $d_3$ get rejected, and apply to their second best hospitals, $h_2$. Since $\bar{q}'_{h_2} = 2$, $h_2$ keeps $d_2$ and $d_3$, and rejects $d_4$ and $d_5$. Finally, $d_5$ applies to her second best hospital $h_3$ and gets tentatively accepted. No further application occurs, so the algorithm terminates.
    The resulting matching is $\mu'$, which is the exact matching produced by the FDA algorithm with order (A) $h_2, h_1$.
\end{example}

\subsection{The efficient allocation and the equivalence}

Our first theorem states the existence of efficient allocations of the regional caps that (weakly) improve the welfare of doctors from the matching produced by the existing mechanism.

\begin{theorem}
    Given some allocation of the regional caps, suppose that the DA algorithm produces a weakly stable matching that is not (constrained) efficient. Then there exist adapted allocations under which the DA algorithm produces a (constrained) efficient weakly stable matching and makes all doctors weakly better off.
\end{theorem}
This result immediately implies the following.
\begin{corollary}
    There exist allocations of regional caps under which the DA algorithm produces a (constrained) efficient weakly stable matching.
\end{corollary}
We prove this existence result using the FDA algorithm. Specifically, we show that the outcome of the FDA algorithm yields a doctor distribution that defines an efficient allocation of the regional caps. Moreover, a stronger result holds: the DA algorithm under this allocation produces the same matching as the FDA algorithm. In short, our main result states that the FDA algorithm is equivalent to the DA algorithm with appropriately adapted capacities. Formally, 

\begin{theorem}\label{T:main}
      Let $\mu^F$ be the matching produced by the FDA algorithm. For each hospital $h$, define $\tilde{q}_{h} \equiv |\mu^F_{h}|$. Let $\mu^D$ denote the matching produced by the DA algorithm with adapted capacities $(\tilde{q}_h)_{h \in H}$. Then, we have $\mu^F = \mu^D$. 
\end{theorem}

The proof, presented in the next subsection, builds on the matching with contracts (\cite{hatfield2005matching}) framework and the optimality of the DA algorithm. For clarity, we define two markets; the \textit{original market} $\mathcal{M}^F$ in which we run the FDA algorithm, and the \textit{shadow market} $\mathcal{M}^D$ in which we run the DA algorithm. 

We begin by formulating both $\mathcal{M}^F$ and $\mathcal{M}^D$ within matching with contracts framework. In this setting, both the FDA algorithm and the DA algorithms can be interpreted as the generalized DA algorithm.

We then show that $\mu^F$ and $\mu^D$ are stable allocation in both markets $\mathcal{M}^F$ and $\mathcal{M}^D$. To facilitate analysis, we use the rationalization of choice functions.
By the doctor-side optimality of the generalized DA algorithm in both markets, we have: (i) $\mu^F \succeq_D \mu^D$ from $\mathcal{M}^F$, and (ii) $\mu^D \succeq_D \mu^F$ from $\mathcal{M}^D$. These imply $\mu^F = \mu^D$, completing the proof. Not only is the argument simple and transparent, it also extends naturally to more general settings, such as \textit{hierarchical} structure (\cite{kamada2018stability}).\footnote{See Section 4 for further details.}

\subsection{Proofs}

In this section, we present the formal proofs for the previous section. We begin with Theorem 2, as the (constrained) efficiency of the FDA algorithm immediately implies the remaining results.

\begin{proof}[Proof of Theorem \ref{T:main}]
    We begin by rephrasing our model using the matching with contracts (\cite{hatfield2005matching}) framework, a method first applied in this context by \cite{kamada2018stability}.

    \textbf{Matching with contracts framework. }
    There are two types of agents: doctors in $D$ and a single agent representing the \textbf{hospital side}. That is, instead of modeling each hospital separately, we treat the entire set of hospitals as one single agent. Thus, there are $|D| + 1$ agents in total. Let $\mathcal{X}^0 \equiv D \times H$ denote the set of all possible \textbf{contracts}, where a contract $x = (d, h) \in \mathcal{X}^0$ represents doctor $d$ being matched to hospital $h$. We assume that each doctor finds any set of more than one contract unacceptable, reflecting the fact that a doctor can be matched to at most one hospital. For each doctor $d$, her preference $\succ_d$ is defined over the set of contracts $(\{d\}\times H) \cup \{\phi\}$, preserving her preference in the original model: $(d, h) \succ_d (d, h')$ if and only if $h \succ_d h'$, and $(d, h) \succ_d \phi$ if and only if $h \succ_d \phi$. Let $C_d: 2^{\mathcal{X}^0} \to \mathcal{X}^0$ denote the choice function of doctor $d$, which selects her most preferred contract from $X \subseteq \mathcal{X}^0$. Define the aggregate choice function of the doctor side by $C_D(X) \equiv \bigcup_{d \in D} C_d(X)$ for all $X \subseteq \mathcal{X}^0$, i.e., $C_D$ simply collects the choices of all doctors.
    The hospital side has preferences represented by a choice rule $C_H$, which we specify later. Let $X^F, X^D \subseteq \mathcal{X}^0$ denote the set of contracts that correspond to the matchings $\mu^F$ and $\mu^D$, respectively. 
    For notational convenience, for any $X \subseteq \mathcal{X}^0$ and $x \in \mathcal{X}^0$, we write $X + x$ to mean the union $X \cup x$. Similarly, we write $X - x$ to mean the set difference $X \backslash x$.
    
    In this framework, the concept of stable allocation corresponds to the standard notion of stability in two-sided matching.
    
    \begin{definition}[Stable allocation (Hatfield and Milgrom 2005)]
        An allocation $X \subseteq \mathcal{X}^0$ is \textbf{stable }if $C_D(X) = C_H(X) = X$ and for any $x \in C_D(X + x) $, we have $x \notin C_H(X +x)$.
    \end{definition}
    
    When $C_H$ satisfies some certain conditions\footnote{Specifically, the following three conditions: A choice function $C$ satisfies the \textbf{law of aggregate demand} if for any set of contracts $X', X \subseteq \mathcal{X}^0$, 
    \begin{align*}
        X' \subseteq X \Rightarrow |C(X')| \leq |C(X)|.
    \end{align*}
    $C$ is \textbf{substitutable} if for all $x, x' \in X$ and $Y \subseteq X\subseteq\mathcal{X}^0$, 
    \begin{align*}
        x \notin C(X \cup \{x\}) \Rightarrow x \notin C(X \cup \{x, x'\}).
    \end{align*}
    $C$ satisfies \textbf{irrelevance of rejected contracts} if
    \begin{align*}
        x \notin C(X \cup \{x\}) \Rightarrow C(X) = C(X \cup \{x\}). 
    \end{align*}
    } then the generalized DA algorithm is known to produce a doctor-optimal stable allocation, which is the counterpart of the doctor-side optimal matching. 
    \begin{definition}[Doctor-optimality]
        A stable allocation $X\subseteq \mathcal{X}^0$ is \textbf{doctor-side optimal} if for any other stable allocation $X' \neq X$, we have $X \succeq_D X'$. 
    \end{definition}
    
    Kamada and Kojima (2018) showed in their more generalized model\footnote{The gerenralized setting has a regional structure called \textit{hierarchy}. See Section 4 for more details.} that the preferences of hospitals and those of regions (that can be constructed from the order of hospitals in each region) can be aggregated into a choice function of the hospital side $C_H$ which satisfies the law of aggregate demand and substitutability, and the (generalized) FDA algorithm\footnote{The generalized FDA algorithm is defined for the setting with hierarchy. When the hierarchical structure reduces to our standard regional constraint, the generalized FDA algorithm reduces to the FDA algorithm.} proceeds exactly the same way as the cumulative offer process in the associated matching with contracts model.\footnote{See Proposition 1 of Kamada and Kojima (2018).} Hence, they proved that $X^F$ is a doctor-optimal stable allocation in the original market $\mathcal{M}^F$. Furthermore, it is almost obvious that $X^D$ is a doctor-optimal stable allocation in the shadow market $\mathcal{M}^D$. 
    
    Next, we show that $X^F$ and $X^D$ are stable allocations in the shadow market $\mathcal{M}^D$ and the original market $\mathcal{M}^F$, respectively. (Notice that stability is a concept that is defined according to the markets.) This completes the proof, since from the optimality of the allocations in each market, we have $X^F \succeq_D X^D$ and $X^D \succeq_D X^F$, which together imply $X^F = X^D$.

\textbf{Rationalizing a choice function.}
     Let $\mathcal{X}\subseteq \mathcal{X}^0$ denote the set of all possible acceptable contracts. A choice function $C$ is said to be \textbf{rationalized} by a utility function $f : 2^{\mathcal{X}} \to \mathbb{R}$ if, for any $Y \subseteq \mathcal{X}$, 
     \begin{align*}
         f(C(Y)) > f(Y') \text{ for every } Y' \subseteq Y \text{ and }Y' \neq C(Y).
     \end{align*}
     
     We rationalize the choice function of the hospital side in two markets.
     For any set of contracts $Y\in \mathcal{X}$, define $\xi(Y)$ as the \textbf{distribution} of the associated matching. This is a vector with $|H|$ elements, where each element specifies the number of doctors in each hospital. For example, if we have only two hospitals $h_1, h_2$ and the set of contracts is $Y \equiv \{(d_1, h_1),(d_2, h_1), (d_3, h_2)\}$, then $\xi(Y) = (2, 1)$. Let $C^D_H$ denote the choice function of the hospital side in the shadow market $\mathcal{M}^D$ (and let $C_H$ denote the choice function of the hospital side in the original market $\mathcal{M}^F$ as before). Note that, the FDA algorithm first fixes some order of the hospitals in each region. This means that, the choice function of the hospital side in the original market has a nature in which, given the applications, first it determines its most preferable distribution and later specifies the identity of the doctors in each hospital according to the preference of each hospital. Thus, we can define two functions $f^F, f^D: 2^\mathcal{X} \to \mathbb{R}$ that rationalize the choice functions $C_H$ and $C^D_H$, respectively, using functions $f_H : 2^\mathcal{X} \to \mathbb{R}$, $f_R, g^F, g^D: \mathbb{Z}^{|H|} \to \mathbb{R}$ and a small positive number $\epsilon$:
    \begin{align*}
        &f^F(Y) = g^F(\xi(Y)) + f_R(\xi(Y)) + \epsilon f_H(Y),\\
        &f^D(Y) = g^D(\xi(Y)) + f_H(Y).
    \end{align*}
    
    Here, $g^F$ and $g^D$ governs the \textbf{feasibility}; $g^F(\xi(Y)) = 0$ if $\xi(Y)$ is feasible in $\mathcal{M}^F$ (and similarly, $g^D(\xi(Y)) = 0$ if $\xi(Y)$ is feasible in $\mathcal{M}^D$,) and $g^F(\xi(Y)) = - \infty$ if $\xi(Y)$ is infeasible in $\mathcal{M}^F$  (and $g^D(\xi(Y)) = - \infty$ if $\xi(Y)$ is infeasible in $\mathcal{M}^D$,) meaning that some of the regional caps or the hospital capacities are violated. $f_R$ governs the preference of the regions in the original market $\mathcal{M}^F$ and $f_H$ governs the preference of the hospitals. Notice that when we have two feasible sets of contracts $Y, Y' \in \mathcal{X}$ with $Y' \subseteq Y$, the definition of the FDA algorithm implies that $f_R(\xi(Y)) > f_R(\xi(Y'))$, that is, the hospital side wants to hold larger set of contracts, given all the contracts are acceptable. Also, note that although we define the choice function only for inputs of the subset of acceptable contracts, whenever we receive the set of contracts with some unacceptable contracts, we can exclude them in choosing the contracts. Now, we show $X^D \succeq_D X^F$ and $X^F \succeq_D X^D$.

\textbf{Step 1. }
    First, we show $X^D \succeq_D X^F$, that is, $X^F$ is a stable allocation in the shadow market $\mathcal{M}^D$. Take any contract $x \in \mathcal{X}\backslash X^F$ such that $x \in C_D(X^F + x)$. Since $X^F$ is a (doctor-side optimal) stable allocation in the original market $\mathcal{M}^F$, we have $C_H(X^F + x) = X^F$ by the definition of stability. This implies that (i) among any $X' \subseteq X^F + x$, $\xi(X^F)$ maximizes $f_R$ (since $\epsilon$ is a small enough positive number,) and (ii) among any $X' \subseteq X^F + x$ with the same distribution of $\xi(X') = \xi(X^F)$, $X^F$ maximizes $f_H$. The definition of the shadow market $\mathcal{M}^D$ implies that any $X' \subseteq X^F + x$ with distribution $\xi(X') > \xi(X^F)$ is not feasible while $\xi(X^F)$ is feasible. Hence, from (ii), we obtain
    \begin{align*}
    {\text{arg max}}_{X' \subseteq X^F + x} f^D(X') = X^F,   
    \end{align*}
    which implies $C^D_H(X^F + x) = X^F$ for any $x \in C_D(X^F + x)$. This means that $X^F$ is a stable allocation in the shadow market $\mathcal{M}^D$ as desired.
    The optimality of $X^D$ in the shadow market $\mathcal{M}^D$ implies 
    \begin{align}\label{eqa:1}
        X^D \succeq_D X^F.
    \end{align}
    By the rural hospital theorem, Equation (\ref{eqa:1}) implies
    \begin{align}\label{eqa:2}
        \xi(X^D) = \xi(X^F).
    \end{align}

\textbf{Step 2.}
    Next, we show $X^F \succeq_D X^D$, that is, $X^D$ is a stable allocation in the original market $\mathcal{M}^F$. Take any contract $\tilde{x} \in \mathcal{X}\backslash X^D$ such that $\tilde{x} \in C_D(X^D + \tilde{x})$. 

    \begin{itemize}
        \item[(i)] If $\tilde{x} \in X^F$, Equation (\ref{eqa:1}) and Equation (\ref{eqa:2}) imply that there exists some contract $x' \in X^D\backslash X^F$ such that $X^F + x' \succ_D X^F$. Since $x' \notin X^F$ and $\tilde{x} \notin X^D$, we have $\xi(X^F + x') = \xi(X^D + \tilde{x})$.
        \item[(ii)] If $\tilde{x} \notin X^F$, $\tilde{x} \in C_D(X^D + x)$ implies $X^F + \tilde{x} \succ_D X^F$. Since $\tilde{x} \notin X^F$ and $\tilde{x} \notin X^D$, we have $\xi(X^F + \tilde{x}) = \xi(X^D + \tilde{x})$.
    \end{itemize}

    For any $\tilde{x}$, (i) and (ii) assure the existence of some contract $x' \in \mathcal{X}$ such that $X^F + x' \succ_D X^F$ and $\xi(X^F + x') = \xi(X^D + \tilde{x})$.

    Note that since we assume $X^F + x' \succ_D X^F$, stability in the original market $\mathcal{M}^F$ requires $C_H(X^F + x') = X^F$, or 
    \begin{align*}
        {\text{arg max}}_{X' \subseteq X^F + x'} f^F(X') = {\text{arg max}}_{X' \subseteq X^F + x'} g^F(\xi(X')) + f_R(\xi(X')) + \epsilon f_H(X') = X^F,
    \end{align*}
    which means among any feasible distribution $\xi(X')$ in $\mathcal{M}^F$ such that $\xi(X') \leq \xi(X^F + x')$, the distribution $\xi(X^F)$ (or equivalently, $ \xi(X^D)$) maximizes $f_R$. (Recall that $\epsilon$ is a small enough positive number.) Since we have $\xi(X^F + x') = \xi(X^D + \tilde{x})$, 
        \begin{align}\label{eqa:3}
            {\text{arg max}}_{\xi(X') \leq \xi(X^D + \tilde{x})} g^F(\xi(X'))  + f_R(\xi(X')) = \xi(X^F) = \xi(X^D).
        \end{align}
        
    Furthermore, since we assume 
$\tilde{x} \in C_D(X^D + \tilde{x})$, stability in the shadow market $\mathcal{M}^D$ requires $C_H^D(X^D + \tilde{x})  = X^D$. This implies that for any $X' \subseteq X^D + \tilde{x}$ with the same distribution $\xi(X') = \xi(X^D)$, $X^D$ maximizes the preference of each hospital $f_H$. Mathematically, 
    \begin{align}\label{eqa:4}
        {\text{arg max}}_{X' \subseteq X^D + \tilde{x}\text{ and } \xi(X') = \xi(X^D)} f_H(X') = X^D.
    \end{align}
    From equations (\ref{eqa:3}) and (\ref{eqa:4}), we have 
    \begin{align*}
        {\text{arg max}}_{X' \subseteq X^D + \tilde{x}} f(X') = {\text{arg max}}_{X' \subseteq X^D + \tilde{x}} g^F(\xi(X')) + f_R(\xi(X')) + \epsilon f_H(X') = X^D,
    \end{align*}
    which means that $C_H(X^D + \tilde{x}) = X^D$ for any $\tilde{x}\in C_D(X^D + \tilde{x})$. Hence, $X^D$ is a stable allocation in the original market $\mathcal{M}^F$ as desired. The optimality of $X^F$ in the original market $\mathcal{M}^F$ implies 
    \begin{align*}
        X^F \succeq_D X^D.
    \end{align*}
From Step 1 and Step 2, we obtain $X^F = X^D$. 
\end{proof}

We proved the equivalence of the FDA algorithm and the DA algorithm with adapted hospital capacities. This result immediately implies the other two results.
\begin{proof}[Proof of Theorem 1]
Let the allocation of the regional caps as $(\tilde{q}_h)_{h\in H}$. Then, Theorem \ref{T:main} and the property of the FDA algorithm implies that the DA algorithm under this allocation produces an efficient matching that weakly Pareto improves the previous matching (for the doctors).
\end{proof}

\begin{proof}[Proof of Corollary 1]
Let the allocation of the regional caps as $(\tilde{q}_h)_{h\in H}$. Then, Theorem \ref{T:main} and the efficiency of the FDA algorithm implies that the DA algorithm produces a (constrained) efficient matching under this allocation.
\end{proof}

\section{Discussion}

This paper offers a new perspective on matching with regional constraints by framing it as an allocation problem of the regional caps and unraveling the connection to the existing literature. Our main result establishes an equivalence between the FDA algorithm and the DA algorithm with endogenously adapted capacities. This equivalence implies that the weak stability of \cite{kamada2017stability}, a fairness notion defined in the original market with regional constraints, actually coincides with the stability under the standard two-sided matching market defined in the shadow market, providing an additional justification for the weak stability. The optimality of the generalized DA algorithm plays a key role in our proof, which builds on the matching with contracts (\cite{hatfield2005matching}) framework, following \cite{kamada2018stability}. The rationalization of the choice function helps us in simplifying the argument.

Our proof extends naturally to more complex regional constraints. \cite{kamada2018stability} studied \textbf{hierarchical} regional constraints, where any two regions are either nested or disjoint. By fixing regional preferences -- analogous to ranking orders in the simple regional cap setting -- and assuming the choice function of the hospital side satisfies substitutability and the law of aggregate demand, they proved the existence of (constrained) efficient matchings and invented the generalized FDA algorithm. Our equivalence result applies directly in this setting: \textit{the generalized FDA algorithm is equivalent to the DA algorithm with endogenously adapted capacities.}

More broadly, our result extends to some environments with distributional constraints where hospital preferences and constraints can be aggregated into a single choice function of the hospital side satisfying substitutability and the law of aggregate demand. In particular, if a DA-based mechanism first determines the distribution and then assigns doctors accordingly, it is equivalent to the DA algorithm with appropriately adapted capacities.

\bibliographystyle{apalike}
\bibliography{hoge}

\end{document}